\documentclass[12pt]{article}

\usepackage{authblk}
\usepackage{amsmath}
\usepackage{amssymb}
\usepackage{enumerate}
\usepackage{framed}

\newtheorem{theorem}{\bf Theorem}[section]

\newtheorem{proposition}[theorem] {\bf Proposition}
\newtheorem{definition} [theorem] {\bf Definition}

\newtheorem{example}    [theorem] {\bf Example}

\newenvironment{proof}{\noindent\mbox{\textbf{Proof} 
}}{\rm\hspace*{\fill}$\rule{7pt}{7pt}$\vspace{10pt}}

\newcommand{\N}{\mathbb{N}}
\newcommand{\R}{\mathbb{R}}
\newcommand{\Ztwo}{\mathbb{Z}_2}

\newcommand{\subsets}{\Delta}
\newcommand{\Ssubsets}{(S,\subsets)}
\newcommand{\SSsubsets}{\mathcal{S}(S,\subsets)}
\newcommand{\symdif}{\bigtriangleup}
\newcommand{\statespace}{\mathop{\mathcal{S}}}
\newcommand{\uint}{\left<0,1\right>}

\title{\bf Quantum logics that are symmetric-difference-closed}

\author{Dominika Bure\v{s}ov\'{a} and Pavel Pt\'{a}k}

\date{}

\begin{document}

\maketitle

\begin{abstract}
    In this note we contribute to the recently developing study of ``almost
        Boolean'' quantum logics (i.e. to the study of orthomodular partially
        ordered sets that are naturally endowed with a symmetric difference).
    We call them enriched quantum logics (EQLs).
    We first consider set-representable EQLs.
    We disprove a natural conjecture on compatibility in EQLs.
    Then we discuss the possibility of extending states and prove an extension
        result for $\Ztwo$-states on EQLs.
    In the second part we pass to general orthoposets with a symmetric
        difference (GEQLs).
    We show that a simplex can be a state space of a GEQL that has an
        arbitrarily high degree of noncompatibility.
    Finally, we find an appropriate definition of a ``parametrization'' as a
        mapping between GEQLs that preserves the set-representation.
\end{abstract}



\section{Introduction}

We shall mostly deal with the following notion of an enriched quantum logic.
Though the text that follows is purely mathematical, let us note that the
    quantum logic is interpreted, in the algebraic foundation of quantum
    mechanics, as a model of the events of a quantum experiment (see
    e.g.~\cite{Gudder,PtakPulmannova}).\\

\begin{definition}
    Let $S$ be a set, $S\ne\emptyset$, and let $\subsets$ be a collection of subsets
    subject to the following conditions (we write $A\smallsetminus B = \{a\in A \ {\rm and} \ a \notin B\}$):
    \begin{enumerate}[(I.)]
        \item $S\in\subsets$,
        \item if $A,B\in\subsets$ and $\symdif$ stands for the symmetric
            difference on $S$, 
            $A\symdif B=(A\smallsetminus B)\cup(B\smallsetminus A)$,
            then $A\symdif B\in\subsets$.
    \end{enumerate}
\end{definition}

The pair $(S,\subsets)$ is said to be an \emph{enriched quantum logic} (an EQL).
We shall sometimes write $\subsets$ instead of $(S,\subsets)$ when we do not
    need refer to the~underlying set $S$.
It may be noted the the notation of EQL is justified.
The collection $\subsets$ indeed forms a set-representable quantum logic in the
    usual sense
    (the partial ordering is given by inclusion).
It is easy to verify that:, of $F$
    \begin{enumerate}[(I.)]
        \item $S\in\subsets$,
        \item $A\in\subsets$ implies $A'=S\smallsetminus A\in\subsets$
            ($S\smallsetminus A$ is the complement of $A$ in $S$),
        \item $A,B\in\subsets$ and $A\cap B=\emptyset$ implies $A\cup B\in\subsets$.
    \end{enumerate}

Obviously, if $A\subseteq B$ then $B=A\cup(B\cap(S\smallsetminus A))$, and so
    $\subsets$ is orthomodular (i.e. $\subsets$ is a QL).
It is easily seen that $(S,\subsets)$ is a Boolean algebra exactly when for any
    $A,B\in\subsets$ we have $A\cup B\in\subsets$.

\section{Results}

The questions addressed in this paper reflect the state of the art of algebraic
    theory of quantum logic.
Technically, it is a branch of orthomodular structures---partially ordered sets with 0 and 1 endowed with an orthocomplementation and subject to the orthomodular law (see
    e.g.~\cite{DeSimoneNavaraPtak, DorferDvurecenskijLanger, Ovchinnikov, MatousekPtak, Su, Voracek}).

The first result concerns the notion of compatibility in EQLs.
The notion of compatibility captures, in theoretical form, the common
    measurability of quantum events~\cite{Gudder, PtakPulmannova}.
Recall that a subset $M$ of an EQL $(S,\subsets)$ is said to be compatible if
    there is a Boolean subalgebra $B$ of $\subsets$ such that $M\subset B$.
Let us say that $(S,\subsets)$ is \emph{compatibility regular} (\emph{compreg})
    if the following implication holds true.
If $A=\{A_1, A_2, \dots, A_n\}\subset\subsets$, and if any subset of $A$ with
    strictly less than $n$ elements is compatible, then so is $A$.
It is known that $(S,\subsets)$ is compreg provided $\subsets$ is a lattice
    (see e.g.~\cite{PtakPulmannova}).
But there are several compreg non-lattice QLs, too. For instance such are several finite set-representable quantum logics (see
    e.g.~\cite{Harding} and logics of splitting subspaces of prehilbert spaces~\cite{PtakWeber}).
Since EQLs are richer than general QLs, it is conceivable that EQLs are compreg.
However,

\begin{theorem}\label{th:2.1}
    Let $S=\{1,2,\dots,2^{n-1},2^n\}$, $n\ge 3$.
    Then there is such an EQL $(S,\subsets)$ that is not compreg.
\end{theorem}

\begin{proof}
    Let $\subsets$ be the collection of all subsets of $S$ that have an even
        number of elements.
    The $(S,\subsets)$ is an EQL.
    Let $i\colon\{0,1\}^n\to S$ be an isomorphism and let $\tilde{A}_i$, 
        $i\le n$, be the subsets of $\{0,1\}^n$ of all elements whose $i$-th
        coordinate is $1$.
    Let $A_i=i(\tilde{A}_i)$, $i\le n$.
    The the collection $\{A_i\mid i\le n\}=A$ is not compatible in $(S,\subsets)$.
    Indeed, $\bigcap_{i\le n}A_i$ is a singleton, $\{p\}$, and
        $\{p\}\not\in\subsets$.
    However, each subset $B$, $B\subset A$, with less than $n$ elements is
        compatible in~$\subsets$.
    This is easy to see since the atoms of the Boolean algebra of subsets of
        $S$ generated by $B$ consist of the intersections of the elements of
        $B$ or their complements.
    Because the number of sets of $B$ is strictly less than $n$, all these
        atoms belong to $\subsets$.
\end{proof}

The next considerations take up questions on states on EQLs (see
    e.g.~\cite{Ptak,DeSimoneNavaraPtak,MatousekPtak}).
The basic definition reads as follows.

\begin{definition}\label{def:2.2}
    Let $(S,\subsets)$ be an EQL and let $s\colon\subsets\to\left<0,1\right>$
        be a mapping.
    Then $s$ is said to be a \emph{state} if
        the following three conditions are satisfied:
        \begin{enumerate}
            \item $s(S)=1$,
            \item $s(A\cup B)=s(A)+s(B)$
                provided
                that $A\cap B=\emptyset$,
            \item $s(A\symdif B)\le s(A)+s(B)$.
        \end{enumerate}
\end{definition}

Let us denote by $\SSsubsets$ the sets of all states on
    $\Ssubsets$.
The following theorem adds to the results of~\cite{DeSimoneNavaraPtak}.

\begin{theorem}\label{th:2.3}~
    \begin{enumerate}
        \item If $s\in\SSsubsets$ then for each $A,B\in\subsets$ the following
            inequality holds:
            $\left|s(A)-s(B)\right|\le s(A\symdif B)$.
        A consequence: If $s(B)=0$, then $s(A)=s(A\symdif B)$ for any
            $A\in\subsets$.
        \item The set $\SSsubsets$ is a convex and compact subset of the
            topological linear space $\R^\subsets$.
        A consequence: If $\Ssubsets$ is a Boolean subalgebra of
            $(S,\tilde{\subsets})$ and $s$ is a state on $\Ssubsets$, then $s$
            can be extended over $(S,\tilde{\subsets})$ as a state.
        \item Suppose that $s$ is a subadditive state on $\Ssubsets$ ($s$ is a
            subadditive state if $s$ satisfies the conditions 1. and 2. of
            Def.~\ref{def:2.2} and if for any $A,B\in\subsets$ there is a set
            $C\in\subsets$, $A\cup B\subset C$ with $s(A)+s(B)\ge s(C)$).
        Then $s$ is a state.
        Further, if $\Ssubsets$ is not Boolean and $\mathop{\mathrm{card}}S$ is
            at most countable, then there is a state that is not a subadditive
            state.
    \end{enumerate}
\end{theorem}

\begin{proof}~
    \begin{enumerate}
        \item It suffices to show that $s(B)-s(A)\le s(A\symdif B)$.
            Consider the sets $A$ and $B'$.
            Then $s(A)+s(B')\ge s(A\symdif B')$.
            Since $A\symdif B'=(A\symdif B)'$, we obtain
                \begin{multline*}
                    s(A)+s(B')
                    =
                    s(A)+1-s(B)
                    \ge
                    s(A\symdif B)
                    \\
                    =
                    s\big((A\symdif B)'\big)
                    =
                    1-s(A\symdif B).
                \end{multline*}
            We see that $s(A)-s(B)\ge -s(A\symdif B)$.
            Hence $s(B)-s(A)\le s(A\symdif B)$.
            The rest is obvious.
        \item It is easy to see that if $s,t\in\SSsubsets$ and
            $\alpha\in\uint$, then $\alpha s+(1-\alpha)t\in\SSsubsets$.
            The topological part is a standard consequence of Tychonoff's
            theorem about the product of compact topological spaces being compact: 
            Since $\uint^\subsets = \{f:\subsets \to \uint \}$ understood as the product of copies of the intervals $\uint$ is a compact subspace of $\R^\subsets$ and
            since $\SSsubsets\subset\uint^\subsets$, it suffices to realize
            that a pointwise limit of states is a state.
            This is easy to see since $s$ is defined by a non-sharp inequality.
            Hence $\SSsubsets$ is a closed subspace of $\uint^\subsets$ and
            therefore $\SSsubsets$ is compact, too.

            The consequence on the extensions has been proved
            in~\cite{DeSimoneNavaraPtak} by an ad~hoc method.
            We would like to show another (simpler) way of proving the result.
            We would use the Krein-Milman theorem: Each state $s$ on
            $\Ssubsets$ belongs to the closure of the convex hull of pure
            states.
            Since pure states on the Boolean algebra $\Ssubsets$ are precisely
            the two-valued states, 
            and the closure of the states on $\Ssubsets$ is taken over to
            $(S,\tilde{\subsets})$,
            it suffices to be able to extend only the
            two-valued states.
            It is well known from the theory of Boolean algebras that one only
            needs to extend $s$ on finite Boolean subalgebras of $\Ssubsets$.
            This is easy---if $B$ is a finite Boolean subalgebra of
            $\Ssubsets$ and $A_1,A_2,\dots,A_n$ are its atoms, there is
            precisely one $A_i$, $i\le n$, that $s(A_i)=1$.
            If we pick a point $a_i$, $a_i\in A_i$, and take the Dirac state
            $d_i$ defined by the point $a_i$ ($d_i(D)=1$ exactly when $a_i\in
            D$, otherwise $d_i(D)=0$), then $d_i$ extends $s$ on $B$ over
            $(S,\tilde{\subsets})$.
        \item Obviously, $s(A)+s(B)\ge s(C)\ge s(A\symdif B)$.
            Further, if $S$ is at most countable and $\Ssubsets$ is not
            Boolean, there are two sets $A,B\in\SSsubsets$ such that $A\cup
            B\not\in\subsets$.
            If we take such a partition of unity $f\colon S\to\uint$,
            $\sum_{x\in S}f(x)=1$, that $f(x)=0$ for any $x\in A\cup B$, then
            $f$ defines a state on $\Ssubsets$ that is not a subadditive state.
            (It may be noted in connection with this observation that
            in~\cite{Ptak} the author constructs examples of non-Boolean EQLs
            on which all states are subadditive inventing rather peculiar types
            of states.)
    \end{enumerate}
\end{proof}

Let us now take up $\Ztwo$-valued states on $\Ssubsets$.
In connection with a potential proposition logic on $\Ssubsets$, when $\symdif$
    can be interpreted as a type of XOR operation, an initial question
    to be asked is about the chance for extending these states.
We will address this question in our next result.
Let~us first recall the definition.

\begin{definition}\label{def:2.4}
    Let $\Ssubsets$ be an EQL.
    The mapping $s\colon\subsets\to\{0,1\}$ is said to be a $\Ztwo$-state if
        \begin{enumerate}
            \item $s(S)=1$,
            \item $s(A\symdif B)=s(A)\oplus s(B)$
                (the operation $\oplus$ is adding mod $2$ in the group
                $\Ztwo$).
        \end{enumerate}
\end{definition}

\begin{theorem}\label{th:2.5}
    Let $(S,\subsets')$ be a subEQL of $(S,\subsets)$ and let
        $s\colon\subsets'\to\{0,1\}$ be a $\Ztwo$-state.
    Then $s$ can be extended over $(S,\subsets)$ as a $\Ztwo$-state.
\end{theorem}

\begin{proof}
    We first observe that $(S,\subsets)$ (and $(S,\subsets')$ as well)
        can be viewed as a linear space (and a linear subspace) over the field
        $\Ztwo$ (this has been observed in~\cite{Ovchinnikov} in connection
        with the game Nim).
    Indeed, if we view the sets of $\subsets$ as their characteristic functions
        with values is $\Ztwo$, then the sets of $\subsets$ can be identified with a
        linear space---if $\chi_1,\chi_2$ are characteristic functions of $A_1,A_2$
        then $\chi_1\oplus\chi_2$ is a characteristic function of $A_1\symdif A_2$.
    Moreover, in this interpretation a $\Ztwo$-state on $(S,\subsets')$ is a
        linear form.
    This form can be extended over $(S,\subsets)$ by the standard method of
        linear algebra.
    The~extension then becomes the $\Ztwo$-state extension over $\Ssubsets$.
\end{proof}

Let us conclude the section on states by an illustrating example.
It shows that even a state that is simultaneously a $\Ztwo$-state (a candidate
    for a ``hidden  variable'' in EQLs) cannot be obtained as a restriction of a
    Boolean state.

\begin{example}\label{ex:2.6}
    Let $S=\{1,2,\dots,9,10\}$.
    Let
        \begin{align*}
            A & = \{2,3,4,5\}, \\
            B & = \{4,6,8,9\}, \\
            C & = \{1,2,4,8\}, \\
            D & = \{4,5,6,7\}.
        \end{align*}
    Let $(S,\subsets)$ be generated by $A$, $B$, $C$ and $D$.
    Then there is a~two-valued state $s$ on $\subsets$ that is also a $\Ztwo$-state and,
        moreover, $s$ cannot be extended as \mbox{a~two-valued} state over $\exp S$ while $s$
        \underline{can} be extended over $\exp S$ as a $\Ztwo$-state.
\end{example}

    Indeed, set 
        $s(S)=1$,
        $s(A)=0$,
        $s(B)=1$,
        $s(C)=1$,
        $s(D)=1$,
        $s(A\symdif B)=1$,
        $s(A\symdif C)=1$,
        $s(A\symdif D)=1$,
        $s(B\symdif C)=0$,
        $s(B\symdif D)=0$,
        $s(C\symdif D)=0$,
        $s(A\symdif B\symdif C)=0$,
        $s(A\symdif B\symdif D)=0$,
        $s(A\symdif C\symdif D)=0$,
        $s(B\symdif C\symdif D)=1$,
        $s(A\symdif B\symdif C\symdif D)=1$.
    The values of $s$ on the complements are determined by the additivity of
        $s$.
    It can be easily checked that $\Ssubsets$ is correctly defined (it is a
        lattice and it has $32$ elements) and that $s$ is a state as well as a
        \mbox{$\Ztwo$-state}.
    Let us first show that $s$ cannot be extended over $\exp S$.
    Let us argue by a~contradiction.

    If $t$ is such an extension, $t$ must be given by a partition of unity
        $f\colon S\to\uint$ such that $t(X)=\sum_{x\in X}f(x)$.
    We have $t(B\symdif C)=t\big(\{1,2,6,9\}\big)=0$ and, also,
        $t\big(\{1,2,5,6,8,9\}\big)=1$.
    It follows that
        $t\big(\{5,9\}\big)=f(5)+f(9)=1$.
    But $f(5)=t(5)=t\big(\{2,3,4,5\}\big)=0$
        and therefore $f(5)= t\big(\{5\}\big)=1$.
    However, $t\big(\{5\}\big) = t\big(\{2,3,4,5\}\big)=0$, a contradiction.
    (Out of separate interest,
        since $s$ is also a $\Ztwo$-state, it must be expressed by means of the
        points of $S$ (Th.~\ref{th:2.5}).)
    Indeed, if we set
        $$v\big(\{4\}\big) = v\big(\{5\}\big) = v\big(\{6\}\big) = v\big(\{9\}\big)=1$$
        and
        $$v\big(\{1\}\big) = v\big(\{2\}\big) = v\big(\{3\}\big) =
        v\big(\{7\}\big) = v\big(\{8\}\big) = v\big(\{10\}\big) = 0,$$
        we see that $v$ extends $s$ as a $\Ztwo$-state on $\exp S$.

In the second part of the paper we pass to general quantum logics with a
    symmetric difference.
They stand to EQLs like the general quantum logics stand to the
    set-representable ones (see~\cite{HrochPtak,MatousekPtak}).

\begin{definition}
    Let $P=(X,\le,\phantom{\,}^\bot,0,1,\symdif)$, where $P=(X,\le,\phantom{\,}^\bot,0,1)$
        is a orthocomplemented partially ordered set and $\symdif\colon X^2\to X$
        is a binary operation.
    Then $P$ is said to be a \emph{generalized enriched quantum logic} (GEQL)
        if $P$ satisfies the following conditions:
        \begin{align}
            & x\symdif(y\symdif z)=(x\symdif y)\symdif z,  \label{GEQL1}\tag{GEQL1} \\
            & x\symdif 1=x^\bot, {\enskip} 1\symdif x=x^\bot,        \label{GEQL2}\tag{GEQL2}   \\
            & x\le z, {\enskip} y\le z {\enskip} \Rightarrow {\enskip} x\symdif y\le z.  \label{GEQL3}\tag{GEQL3} 
        \end{align}
\end{definition}

It can easily be proved that a GEQL is an EQL if the GEQL
    is set-representable and the symmetric difference is formed set-theoretically (\cite{MatousekPtak}).
Also, the GEQL is an appropriate notion within quantum logics since any GEQL
    is orthomodular.
To indicate a direct proof of the last fact (a detailed analysis can be found (\cite{MatousekPtak})), suppose that $P$ is a GEQL and 
    $x\le y$, $x,y\in P$.

We want to check the orthomodular law: $y=x\lor(y\land x^\bot)$.
We have 
\begin{multline*}
y\land(x\lor(y\land x^\bot))^\bot=
y\land(x^\bot\land(y\land x^\bot)^\bot)\\
=(y\land x^\bot)\land(y\land(y\land x^\bot)^\bot)
\leq(y\land x^\bot)\land(y\land x^\bot)^\bot=0.
\end{multline*}
To show that this implies the orthomodularity, write $z=(x\lor(y\land x^\bot))$.
We~have $z\le y$ and $y\land z^\bot=0$.
By applying~\eqref{GEQL3}, 
    $y\land z^\bot=y\symdif z=0$ and, further, by~\eqref{GEQL1} and the obvious identity $y\symdif y = 0$ we see that
    $$z=z\symdif 0=z\symdif(y\symdif y)=(z\symdif y)\symdif y=0\symdif y=y.$$
This proves the orthomodularity.

Our first result on GEQLs concerns the state space.
Recall that a state on GEQL $P$ is a mapping $s\colon P\to\left<0,1\right>$
    such that:
    \begin{enumerate}
        \item $s(1)=1$,
        \item $s(a\cup b)=s(a)+s(b)$ provided $a\le b^\bot$,
        \item for any $a,b\in P$, $s(a\symdif b)\le s(a)+s(b)$.
    \end{enumerate}
It can again be shown that for any GEQL $P$, the state space
    $\statespace(P)$ of $P$ is a compact convex set (in the natural affine
    structure and pointwise topology).
It seems conjecturable that each compact convex set in a topological linear
    space can be seen as a state space of GEQL.
The following partial result in this line supports this conjecture (we add the
    requirement on non-compatibility---there are QEQLs that have a ``standard''
    state space though a high degree of non-compatibility).

\begin{theorem}\label{th:2.8}
    Let $K$ be a simplex in $\R^n$, $n\in\N$.
    Let $k$ be a cardinal number.
    Then there is a GEQL, $P$, that contains $k$ non-compatible pairs and
        that has $\statespace(P)$ affinely homeomorphic with $K$.
\end{theorem}

\begin{proof}
    There is a GEQL, $\mathcal{P}$, such that $\mathop{\mathcal{S}}(\mathcal{P})=\emptyset$ (see~\cite{PtakVoracek}).
    Then the Cartesian product $\mathcal{P}^k$ obviously has at least $k$ non-compatible pairs 
    (GEQLs form a quasivariety~\cite{MatousekPtak}
        and therefore they are closed under the formation of Cartesian products).
    Moreover, $\mathop{\mathcal{S}}(\mathcal{P}^k)=\emptyset$.
    Indeed, $\mathcal{P}$ can be naturally embedded in $\mathcal{P}^k$ and if there is a state on $\mathcal{P}^k$, there is a state on $\mathcal{P}$.
    Further, $\mathcal{P}^k\times\{0,1\}$, where $\{0,1\}$ is viewed as a GEQL, possesses exactly one state---it suffices 
    to set $s(a,0)=0$ and $s(a,1)=1$, $a\in\mathcal{P}^k$.
    Write $Q=P^k\times\{0,1\}$ and consider $Q^{n+1}$, where ${n+1}$ is the number of the extreme points of $K$.
    Then $Q^{n+1}$ has also precisely ${n+1}$ pure states and therefore $\statespace(Q^{n+1})=\statespace(K)$.
\end{proof}

Our final result discusses a way of ``parametrization'' of GEQLs.
In view of the distinguished position of EQLs among GEQLs, it is natural to ask
    which kind of morphisms from
EQLs onto GEQLs guarantees that the images of EQLs are again EQLs.
The investigation of~\cite{PtakWright} indicates that a stronger condition on
    compatibility must be introduced.
We employ the following definition.

\begin{definition}
    Let $P,Q$ be GEQLs and let $p\colon P\to Q$ onto $Q$ be a mapping.
    Let $P$ be GEQL-isomorphic to an EQL (i.e., let $P$ be set-representable).
    Then $p$ is said to be a \emph{parametrization} of $Q$ if
        \begin{enumerate}
            \item $f(1)=1$,
            \item $f(a\lor b)=f(a)\lor f(b)$ provided that $a\le b^\bot$,
            \item if $a\le b$ then $f(a)\le f(b)$,
            \item $f(a\symdif b)=f(a)\symdif f(b)$,
            \item if $a\le b$ in $Q$ and $a=p(A)$, $b=p(B)$, then $A,B$ are compatible.
        \end{enumerate}
\end{definition}

The requirement that $p$ is a parametrization is rather strong.
If e.g. $Q$ is a GEQL and $B$ is a Boolean algebra, then the projection of $Q\times B$ onto $Q$
    is a parametrization, and from this example some other parametrization can be
    constructed.

Before we formulate the result, let us state a simple characterization of EQLs among GEQLs
(see~\cite{MatousekPtak}).

\begin{proposition}\label{prop:GEQL.isomorphic.to.EQL}
    Let $Q$ be a GEQL.
    Then $Q$ is GEQL-isomorphic to an EQL if there is a mapping 
        $m\colon Q\to\exp M$ for a set $M$ such that the following conditions
        are fulfilled:
        \begin{enumerate}
            \item $x\le y^\bot\Leftrightarrow m(x)\le M\smallsetminus m(y)$,
            \item $m(x\symdif_Q y)=m(x)\symdif m(y)$, where the right hand side is the set
                symmetric difference.
        \end{enumerate}
\end{proposition}

We can spell out the following result.

\begin{theorem}
    Let $p\colon P\to Q$ be a parametrization between two GEQLs.
    Let $P$ be GEQL-isomorphic to an EQL.
    Then $Q$ is GEQL-isomorphic to an~EQL, too.
\end{theorem}

\begin{proof}
    We will use Prop.~\ref{prop:GEQL.isomorphic.to.EQL}.
    Let $P$ be GEQL-isomorphic to an EQL.
    So we can assume $P=(S,\subsets)$.
    Consider the Boolean algebra $\exp S$ of all subsets of $S$.
    Let $e\colon\subsets\to\exp S$ be the natural embedding.
    Write $I=\{A\in\subsets\mid p(A)=0\}$.
    Then $I$ is compatible in $P$ and therefore can be viewed as a subset of a Boolean subalgebra of $\exp S$.
    Write $J=\{D\in\exp S\mid D\subset A\}$ for some $A\in I$.
    Then $J$ is an ideal in $\exp S$ and therefore $B=\exp S/J$ is a Boolean algebra.
    Let us denote the factor mapping by $n$, $n\colon\exp S\to B$.
    Obviously, $B$ is also a collection of subsets of a set by the Stone theorem.

    Let us denote by $m\colon Q\to B$ the following mapping.
    For each $a\in Q$, let us choose an $A$, $A\in\Delta$, such that $p(A)=a$,
        and set $m(a)=m(p(A))=n(e(A))$.
    We are going to show that $m$ is correctly defined.
    Indeed, suppose that $a=p(A)=p(B)$.
    Then $A,B$ are compatible in $(S,\subsets)$ and therefore 
        $(A\cap B')\cup(B\cap A')\in\subsets$
        (we write $B'$ instead of $S\smallsetminus B$).
    By a Boolean calculus, 
        $$p\Big((A\cap B')\cup(B\cap A')\Big)=0$$ and this implies
        that $(A\cap B')\cup(B\cap A')\in J$.
    It means that $n(e(A))=n(e(B))$ and so $m$ is defined correctly.

    It is not difficult to show that $m\colon Q\to B$ fulfils all conditions of
        Prop.~\ref{prop:GEQL.isomorphic.to.EQL}.
    Let us for instance check that if $a,b\in Q$, $a\le b'$, then
        $m(a)\le\big(m(b)\big)'$,
        the~other conditions can be verified analogously.
    Suppose $a=p(A)\le b'=p(B')$.
    Since $A,B'$ are compatible and $p(B)\le \big(p(A)\big)'=p(A')$,
        the intersection $B\cap A'$ belongs to $\subsets$
        and $p(B)=p(B\cap A')$.
    We see that
        \begin{multline*}
            m(a)=m(p(A))=n(e(A))\le n(e(B\cap A')')=m(p(B\cap A')') \\
            =m(p(B\cap A'))'
            =m(p(B))'=m(b')=(m(b))'.
        \end{multline*}
\end{proof}

Let us in conclusion shortly comment on the position of ``enriched quantum logics'', EQL (and their generalized forms GEQL) considered in this paper within the quantum logics, QL. The study of QL (QL = orthomodular poset) has taken place in the traditional projection area (\cite{BirkhoffNeumann, Hamhalter, PtakWeber}, etc.), in the specific orthomodular combinatorics (\cite{Greechie, Voracek, Svozil}, etc.) and in the set-representable structures close to Boolean algebras (\cite{DeSimoneNavaraPtak, Harding, Ptak}, etc.). The areas overlap, of course (\cite{DvurecenskijPulmannova, Gudder, Redei}, etc.). In this paper we want to further contribute to the third direction---we ``enrich'' the set-representable QL with an XOR type operation, obtaining thus a kind of ``almost Boolean'' QL (\cite{HrochPtak, MatousekPtak, Su}, etc.). This stands to reason as long as the (non)compatibility relation is concerned (Th.~\ref{th:2.1}) and it opens new types of investigations (Th.~\ref{th:2.3}, Th.~\ref{th:2.5} and Th.~\ref{th:2.8}). The further research may bring interesting results in QL, in particular in the investigation of the state space (is an arbitrary compact convex set the state space of a GEQL?) and, possibly, in a generalized QL-line of ``mathematical logic'' (having proved Th.~\ref{th:2.3}, can we think of a kind of propositional logic in EQL?).

\section*{Acknowledgement}

The second author acknowledges the support by the Austrian Science Fund
    (FWF):Project I 4579-N and the Czech Science Foundation (GACR):Project
    20-09869L.
    
The final publication is available at Springer via https://doi.org/10.1007/s10773-021-04950-6.

\begin{minipage}[t]{0.5\textwidth}
    \centering
    \textbf{Dominika Bure\v{s}ov\'{a}}\\
    \smallskip
    Czech Technical University in Prague,\\
    Faculty of Electrical Engineering,\\
    Prague, Czech Republic\\
    \smallskip
    \texttt{buresdo2@fel.cvut.cz}
\end{minipage}
\qquad
\begin{minipage}[t]{0.5\textwidth}
    \centering
    \textbf{Pavel Pt\'{a}k}\\
    \smallskip
    Czech Technical University in Prague,\\
    Faculty of Electrical Engineering,\\
    Department of Mathematics,\\
    Prague, Czech Republic\\
    \smallskip
    \texttt{ptak@math.feld.cvut.cz}
\end{minipage}

\end{document}